\documentclass[aps,pra,twocolumn,superscriptaddress]{revtex4-2}

\usepackage{graphicx,subfigure}
\usepackage{amsmath,amssymb,bm,amsthm}
\usepackage{mathtools}
\usepackage{multirow}
\usepackage{textcomp}
\usepackage{float}
\usepackage{xcolor}
\usepackage[normalem]{ulem}
\usepackage{dsfont}
\usepackage{url}
\usepackage{comment}
\usepackage{mathrsfs}
\usepackage{makecell}
\usepackage{array}

\usepackage{tikz}
\usetikzlibrary{shapes, arrows, positioning}

\newcommand{\ave}[1]{\left \langle #1 \right \rangle}

\newcommand{\opa}{\hat{a}}
\newcommand{\opad}{\hat{a}^\dagger}
\newcommand{\opb}{\hat{b}}
\newcommand{\opbd}{\hat{b}^\dagger}

\newcommand{\oph}{\hat{H}}
\newcommand{\opv}{\hat{v}}
\newcommand{\opvd}{\hat{v}^\dagger}

\newtheorem{theorem}{Theorem}

\usepackage{hyperref}
\hypersetup{
    colorlinks,
    citecolor=black,
    filecolor=black,
    linkcolor=black,
    urlcolor=black
}

\begin{document}

\title{Quantum communication over bandwidth-and-time-limited channels}

\author{Aditya Gandotra}
\thanks{These two authors contributed equally}
\email{adityag@uchicago.edu}
\author{Zhaoyou Wang}
\thanks{These two authors contributed equally}
\email{zhaoyou@uchicago.edu}
\author{Aashish A. Clerk}
\email{aaclerk@uchicago.edu}
\author{Liang Jiang}
\email{liangjiang@uchicago.edu}

\affiliation{Pritzker School of Molecular Engineering, University of Chicago, Chicago, Illinois 60637, USA}

\date{\today}

\begin{abstract}
    Standard communication systems have transmission spectra that characterize their ability to perform frequency multiplexing over a finite bandwidth.
    Realistic quantum signals in quantum communication systems like transducers are inherently limited in time due to intrinsic decoherence and finite latency, which hinders the direct implementation of frequency-multiplexed encoding.
    We investigate quantum channel capacities for bandwidth-and-time-limited (BTL) channels to establish the optimal communication strategy in a realistic setting. 
    For pure-loss bosonic channels, we derive analytical solutions of the optimal encoding and decoding modes for Lorentzian and box transmission spectra, along with numerical solutions for various other transmissions. 
    Our findings reveal a general feature of sequential activation of quantum channels as the input signal duration increases, as well as the existence of optimal signal length for scenarios where only a limited number of channels are in use.
\end{abstract}

\maketitle

\emph{Introduction}\textemdash
Sending quantum states coherently from one place to another is essential in quantum networks~\cite{kimble2008}.
Matter qubits that store the quantum information are first encoded into photons, the ideal information carrier for long-distance communication.
These photons are then transmitted through a communication channel and measured at the receiving end to recover the original qubits (Fig.~\ref{fig:fig0}).
Quantum capacity, the maximum number of qubits that can be reliably transmitted per channel use, serves as the key performance metric for quantum communication channels.

Most communication systems, such as optical fibers~\cite{cirac_quantum_1997,ritter2012}, microwave links~\cite{wenner2014,campagne-ibarcq_deterministic_2018,magnard2020,zhong2021,storz2023}, and quantum transducers~\cite{han_microwave-optical_2021,lambert2020a,jiang2020,mckenna2020,mirhosseini2020,xu2021,shen2022,tu2022,meesala2024,zhang2018,zhong2020a,wu2021,wang2023,shi2024}, allow simultaneous transmission of photons at different frequencies over a finite bandwidth.
As a result, there are infinite number of modes available for encoding and transmitting qubits.
Frequency multiplexing, for example, encodes quantum information simultaneously in multiple input modes $\opa_{\text{in}}(\omega)$ at different frequencies.
More generally, we can define an input mode as a wave packet
\begin{equation}
    \label{eq:mode_def}
    \hat{a} = \int_{-\infty}^{\infty} f(\omega) \hat{a}_{\text{in}}(\omega) \text{d}\omega ,
\end{equation}
where the mode profile $f(\omega)$ can be controlled experimentally~\cite{pechal2014,wenner2014,campagne-ibarcq_deterministic_2018}.
Quantum information can be encoded simultaneously in multiple input modes $\{\opa_k\}$ with mutually orthogonal mode profiles $\{f_k(\omega)\}$ to achieve higher information rates (Fig.~\ref{fig:fig0}).

Frequency multiplexing is no longer viable in realistic scenarios where quantum signals have a have time-duration $T$ that is is not infinitely larger than relevant time scales of the communication channel.
Constraints on maximum signal time arise from several factors, including internal losses in the matter qubits carrying the signals, fluctuations of the communication channel over long time scales, and additional decoherence from heating effects (e.g., strong optical pumps in optomechanical quantum transducers) which can even disable the quantum channel~\cite{meenehan2015,urmey_stable_2024,mayor_two-dimensional_2024}.
We define bosonic channels with finite bandwidths and signal durations as the bandwidth-and-time-limited (BTL) channels.
It is worth noting that quantum capacity is compatible with BTL channels: the capacity is defined in the asymptotic limit of many channel uses, with each input signal being time-limited. However, quantum capacity of BTL channels remains unexplored.

Even when there are no external constrains requiring finite duration signals, time-limited signals can offer practical advantages. For instance, a finite signal duration reduces latency in quantum networks. In near-term quantum state transfer experiments, where only one or a few input modes can be controlled to encode and transmit qubits, there must be some optimal time duration for the input modes. If the duration is too short, the signal bandwidth exceeds the channel bandwidth, resulting in signal reflection. 
Conversely, for very long signals, quantum capacity using only a few input modes saturates, reducing the information rate per unit time.
Indeed, a key result of our work is is to determine the optimal finite-time encoding when using a small number of input modes. 

\begin{figure}
    \centering
    \includegraphics[width = 0.5 \textwidth]{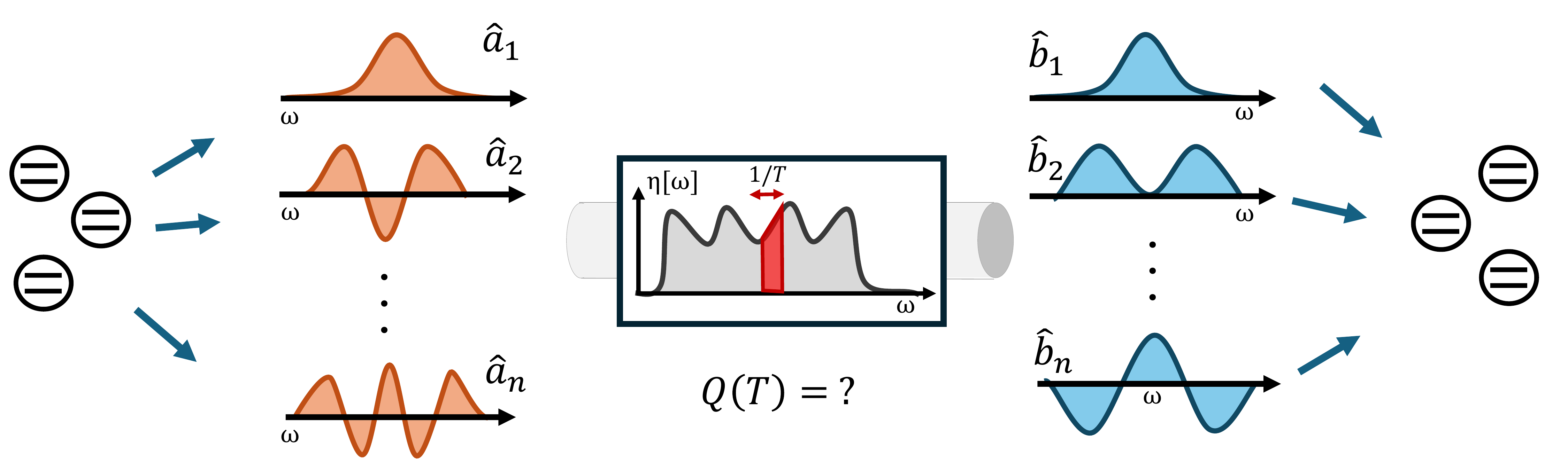}
    \caption{Schematic of sending qubits through a channel with finite bandwidth. We consider the scenario where qubits are encoded into a set of orthogonal photon modes with finite time duration $T$.}
    \label{fig:fig0}
\end{figure}

Determining the quantum capacity of BTL channels with pure loss requires finding the optimal orthogonal input modes and output readout, a nontrivial task due to the nonorthogonality of the output mode profiles.
In frequency multiplexing, the mode profiles are delta functions in the frequency domain which remain orthogonal under any transmission (Fig.~\ref{fig:schematic}(b)).
However, a finite timescale imposes a minimum linewidth on the mode profiles $\{f_k(\omega)\}$, resulting in deformed output mode profiles $\{g_k(\omega)\}$ that are no longer mutually orthogonal (Fig.~\ref{fig:schematic}(c)) due to the frequency-dependent transmission.
While an analogous problem has been studied for restricted spatial optical modes in classical communications~\cite{miller2000communicating,miller2019waves,lupo2012}, the optimal encoding and readout strategies in time-limited scenarios as well as the resulting quantum capacity remain unknown.

In this paper, we study quantum communication in BTL pure-loss channels.
The optimal strategy is to find proper input modes such that the corresponding output mode profiles are mutually orthogonal, which decomposes the multi-mode channel over finite bandwidth into parallel single-mode channels.
We prove that the total quantum capacity of a generic multi-mode pure-loss channel is precisely the sum of quantum capacity of each single-mode channel.
Then, we derive the (total) quantum capacity for a BTL pure-loss channel, which demonstrates step-wise increase of one-way quantum capacity~\cite{wilde2013quantum} as we increase the time duration, associated with individual channels opening up.
Our work also determines the optimal signal duration and input modes for near-term experiments with limited number of controllable channels.
 
\begin{figure}
    \centering
    \includegraphics[width = 0.5 \textwidth]{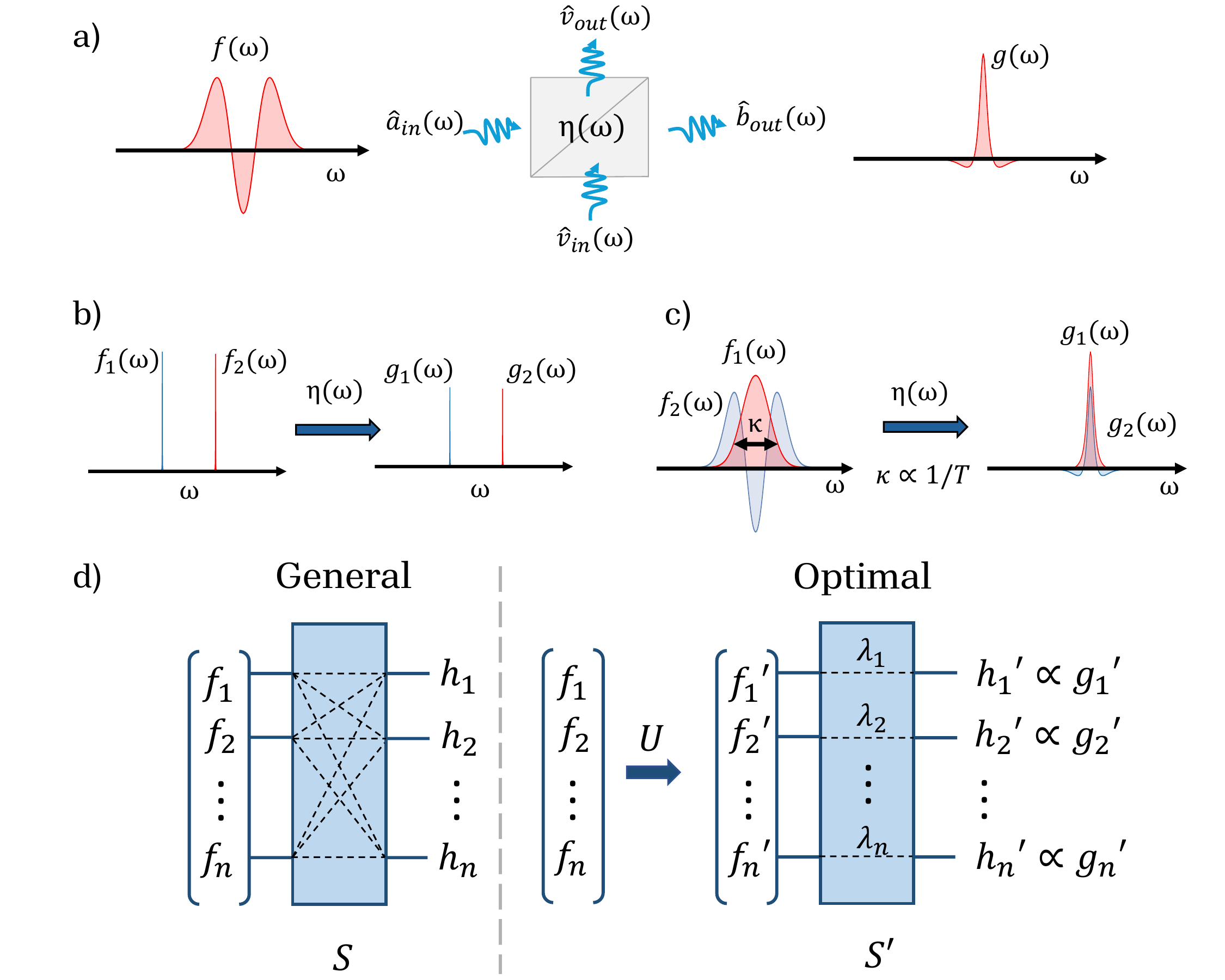}
    \caption{(a) Transmission of a wave packet through a linear scattering channel with transmission spectrum $\eta(\omega)$, transforming input mode profile \(f(\omega)\) to output profile \(g(\omega)\). (b) Infinite-duration signals with delta function mode profiles in the frequency domain transform to delta functions under any transmission. (c) Time-limited orthogonal input modes with minimum linewidth \(\kappa\) do not generally remain orthogonal after transmission through the channel. (d) Schematic of general input and readout modes and the optimal setting in Theorem~\ref{theorem}.}
    \label{fig:schematic}
\end{figure}

\emph{Optimized communication with one input mode}\textemdash
We consider pure-loss channels with frequency-dependent transmission coefficients.
At frequency $\omega$, an input quantum signal $\opa_{\text{in}}(\omega)$ gets converted into an output quantum signal $\opb_{\text{out}}(\omega)$ with the input-output relation (Fig.~\ref{fig:schematic}(a))
\begin{equation}
    \label{eq:channel_IO}
    \opb_{\text{out}}(\omega) = \tau(\omega) \opa_{\text{in}}(\omega) + r(\omega) \opv_{\text{in}}(\omega) ,
\end{equation}
where $\opv_{\text{in}}(\omega)$ is the vacuum input from the environment that adds extra vacuum noise to the output.
The transmission and reflection coefficients $\tau(\omega)$ and $r(\omega)$ satisfy $|\tau(\omega)|^2 + |r(\omega)|^2 = 1$, and the frequency-dependent transmissivity is $\eta(\omega) = |\tau(\omega)|^2$.
The field operators satisfy the usual bosonic commutation relations $[\opa_{\text{in}}(\omega),\opad_{\text{in}}(\omega')] = [\opb_{\text{out}}(\omega),\opbd_{\text{out}}(\omega')] = [\opv_{\text{in}}(\omega),\opvd_{\text{in}}(\omega')] = \delta (\omega -\omega')$.
We assume vacuum states for the environment modes $\opv_{\text{in}}(\omega)$ throughout the paper.

To incorporate constraints on input signals, it is essential to understand the transmission of signal modes with general mode profiles beyond single-frequency encodings.
Consider first the simple case of a single input mode $\opa$ with mode profile $f(\omega)$ (Eq.~(\ref{eq:mode_def})), which is normalized as $\int_{-\infty}^{\infty} |f(\omega)|^2 d\omega=1$.
After the transmission, the unnormalized output profile is $g(\omega) = \tau^*(\omega) f(\omega)$.
For a readout mode $\opb = \int h(\omega) \opb_{\text{out}}(\omega) d\omega$ with a normalized profile $h(\omega)$, we have
\begin{equation}
    \label{eq:one_mode_IO}
    \opb = \ave{g,h} \opa + \sqrt{1-|\ave{g,h}|^2} \opa_\perp ,
\end{equation}
where $\opa_\perp$ is a vacuum mode orthogonal to $\opa$ and $\ave{g,h} \equiv \int_{-\infty}^{\infty} g^*(\omega) h(\omega) \text{d}\omega$
is the inner product between functions.

Since $\opa_\perp$ is the vacuum input from the environment, the quantum channel $\mathcal{E}:\opa \rightarrow \opb$ in Eq.~(\ref{eq:one_mode_IO}) is a bosonic pure-loss channel~\cite{weedbrook_gaussian_2012} with transmissivity $\eta = |\ave{g,h}|^2$.
The quantum capacity of a bosonic pure-loss channel with transmissivity $\eta \in [0,1]$ is
\begin{equation}
    \label{eq:pure_loss_capacity}
    q(\eta) = \max{\left\{0, \log_2\left(\frac{\eta}{1-\eta}\right) \right\}} ,
\end{equation}
which is the maximal number of qubits that can be reliably transmitted per channel use.
Since quantum capacity increases monotonically with $\eta$, the optimal output mode $\opb$ for reading out $\opa$ should maximize $\eta=|\ave{g,h}|^2$. From the Cauchy-Schwarz inequality, we have
\begin{equation}
    \eta = |\ave{g,h}|^2 \leq \ave{g,g} \cdot \ave{h,h} = \ave{g,g} .
\end{equation}
Therefore the optimal output mode $\opb$ is given by $h(\omega) = g(\omega)/\sqrt{\ave{g,g}}$ which leads to the maximal $\eta = \ave{g,g}$. 

\emph{Multi-mode pure-loss channels}\textemdash
More generally, we can encode quantum signals in multiple input modes simultaneously and measure multiple output modes to receive the quantum information.
For input modes $\bm{a} \equiv (\opa_1,\ldots,\opa_n)^T$ with profiles $\{f_1,\ldots,f_n\}$, we can measure output modes $\bm{b} \equiv (\opb_1,\ldots,\opb_n)^T$ with profiles $\{h_1,\ldots,h_n\}$, where $\ave{f_i,f_j} = \ave{h_i,h_j} = \delta_{ij}$. Then we can define a general scattering matrix \(S\) by $S_{lk} = \ave{g_k, h_l}$ with $g_k(\omega) = \tau^*(\omega) f_k(\omega)$ (Fig.~\ref{fig:schematic}(d)).
This leads to a multi-mode quantum channel
\begin{equation}
    \label{eq:multi_mode_channel}
    \mathcal{E}: \bm{a} \rightarrow \bm{b} = S \bm{a} + \cdots \bm{a}_\perp ,
\end{equation}
where $\bm{a}_\perp$ are vacuum modes orthogonal to $\bm{a}$.

\begin{theorem}
    \label{theorem_capacity}
    The quantum capacity $Q(\mathcal{E})$ of a generic multi-mode pure-loss channel $\mathcal{E}$ (Eq.~(\ref{eq:multi_mode_channel})) is $\sum_{k=1}^n q (\sigma_k^2)$, where $\{ \sigma_1,...,\sigma_n \}$ are the singular values of $S$.
\end{theorem}

\begin{proof}
    We can rewrite $\mathcal{E}$ in an equivalent diagonal form with $\bm{a}' = U\bm{a}$ and $\bm{b}' = V\bm{b}$, where the unitary $U,V$ come from the singular value decomposition $S=V^\dagger \Sigma U$. This gives the decomposition $\mathcal{E} = \mathcal{E}_1 \otimes \cdots \otimes \mathcal{E}_n$, where $\mathcal{E}_k: \opa_k' \rightarrow \opb_k'$ is a bosonic pure-loss channel with transmissivity $\sigma_k^2$.
    In general, the quantum capacity of $\mathcal{E}$ is intractable since $\mathcal{E}$ is neither degradable nor anti-degradable~\cite{wolf2007,caruso2008} with some $\sigma_k^2$ greater than 0.5 and others less than 0.5. However, for pure-loss channels, we have an upper bound $Q(\mathcal{E}) \leq Q(\bar{\mathcal{E}})$ from the data processing inequality, where $\bar{\mathcal{E}} = \bar{\mathcal{E}}_1 \otimes \cdots \otimes \bar{\mathcal{E}}_n$ and $\bar{\mathcal{E}}_k$ is a pure-loss channel with transmissivity $\max\{ \sigma_k^2, 0.5 \}$. Since $\bar{\mathcal{E}}$ is degradable, we have $Q(\bar{\mathcal{E}})=\sum_{k=1}^n q (\sigma_k^2)$~\cite{wolf2007}. On the other hand, we also have $\sum_{k=1}^n q (\sigma_k^2) \leq Q(\mathcal{E})$ since we can use the quantum channels $\{\mathcal{E}_k\}$ independently with separable input states to $\mathcal{E}$. This completes the proof.
\end{proof}

Note that for a given set of input modes, the transmissivities $\{\sigma_k^2\}$ depend on the choice of the readout modes. The optimal setting (Fig.~\ref{fig:schematic}(d)) that achieves the maximal transmissivity for every channel $\mathcal{E}_k$ is (see Appendix~\ref{appendix:proof}):
\begin{theorem}
    \label{theorem}
    Given a particular basis $\{f_1,\ldots,f_n\}$ of input modes $\bm{a}=(\opa_1,\ldots,\opa_n)^T$, we can diagonalize the Hermitian matrix $M=U^T D U^*$ with unitary $U$ where $M_{ij} = \ave{g_i,g_j}$. The optimal input modes are $\bm{a}' = U\bm{a}$ with profiles $\{f_1',\ldots,f_n'\}$ whose output profiles $\{ g_1',...,g_n'\}$ are mutually orthogonal.
    The optimal readout modes are $\bm{b}' = (\opb_1',\ldots,\opb_n')^T$ with $h_k' = g_k'/\sqrt{\ave{g_k',g_k'}}$, and the single-mode channels $\{ \mathcal{E}_k:\opa_k' \rightarrow \opb_k' \}$ are bosonic pure-loss channels whose transmissivities $\{\lambda_k=\ave{g_k',g_k'}\}$ are the eigenvalues of $M$.
\end{theorem}

From now on, we only focus on the optimal setting and use $\opa_k,f_k$ to denote the optimal modes and profiles for notational simplicity, rather than $\opa_k',f_k'$.

\emph{Quantum capacity of BTL pure-loss channels}\textemdash
We can apply Theorem~\ref{theorem} to find the optimal input modes under the constraint that the mode profiles only have a finite support over $t\in [-T/2, T/2]$ in the time domain.
Choosing delta functions in the time domain as the basis functions, the Hermitian matrix $M$ in Theorem~\ref{theorem} becomes a linear operator $M(t,t') = \tilde{\eta}(t-t')$, where $t,t' \in [-T/2, T/2]$ and $\tilde{\eta}(t) = \frac{1}{2\pi}\int d\omega \  e^{i\omega t} \eta(\omega)$.
The optimal input mode profiles $\{f_k\}$ in the time domain are eigenfunctions of $M$ satisfying (see Appendix~\ref{appendix:derivation})
\begin{equation}
\label{eigProblem}
    \int_{-T/2}^{T/2} dt'\ \tilde{\eta}(t-t') f_k(t') = \lambda_k f_k(t) ,
\end{equation}
which can be solved numerically by discretizing \(\tilde{\eta}(t)\).

The eigenvalues \(\{\lambda_k\}\) in Eq.~(\ref{eigProblem}) are the transmissivities of the optimal single-mode channels for time-limited quantum signals. Therefore, the time-limited quantum capacity over finite bandwidth with signal duration $T$ is $\sum_k q(\lambda_k)$, and we define the quantum capacity per unit time as the \textit{time-limited quantum capacity rate}
\begin{equation}
    \label{eq:time_limited_capacity}
    Q(T) = \frac{1}{T}\sum_{k=1}^{\infty} q(\lambda_k) .
\end{equation}
Note that in the case where the input-signal duration \(T \rightarrow \infty\), a \textit{continuous-time quantum capacity}~\cite{wang_quantum_2022} for frequency multiplexing has been defined as
\begin{equation}
    \label{eq:maxCap}
    Q^{\text{max}} = \frac{1}{2 \pi}\int_{-\infty}^{\infty} d\omega \  q\left[ \eta(\omega) \right] .
\end{equation}
It can be proved that the time-limited quantum capacity rate $Q(T)$ approaches the continuous-time quantum capacity \(Q^{\text{max}}\) as $T\rightarrow \infty$ (see 
Appendix~\ref{appendix:convergence}).

\begin{figure}
    \centering
    \includegraphics[width = 0.5\textwidth]{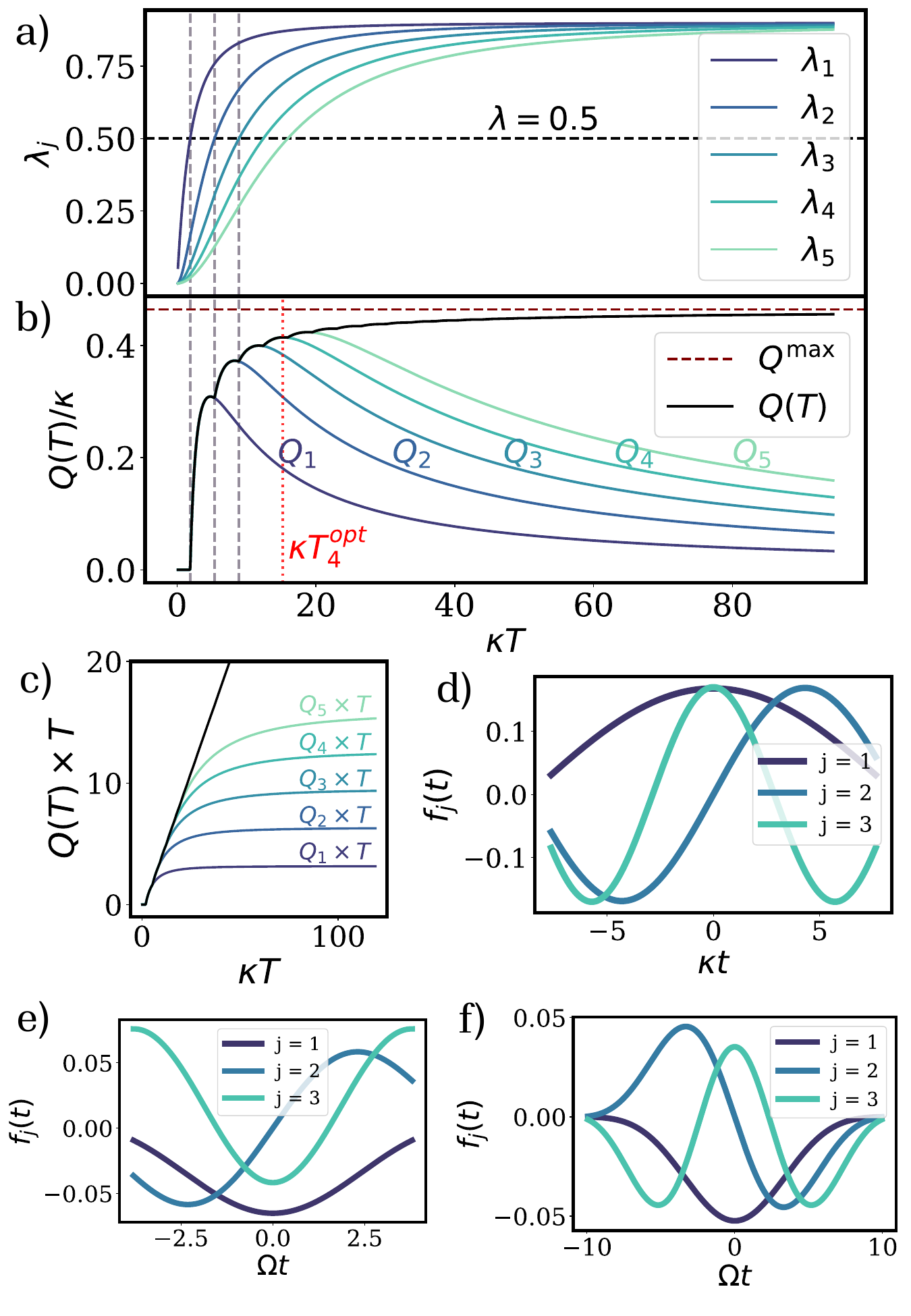}
    \caption{(a) The first five transmissivities \(\lambda_j\) of a Lorentzian transmission with \(\eta_{\text{max}} = 0.9\) for different signal duration $T$. (b) Time-limited quantum capacity rate $Q(T)$ (black), where the discontinuities correspond to transmissivities \(\lambda_j\) crossing 0.5 since a new channel opens. The colored lines represent the contributions $Q_k(T)$ of the first \(k\) channels. (c) Time-limited quantum capacity \(Q(T) \times T\) (black) and contributions $Q_k(T)\times T$ from the first $k$ channels (colored). (d) The first three optimal input modes for \(T = T^{\text{opt}}_4 \approx 15.261/\kappa\). Here $T^{\text{opt}}_4$ maximizes the communication rate for four multiplexed modes. (e-f) The first three optimal input modes of a box transmission with height $\bar{\eta}=0.85$ for (e) \(T = T_2^{\text{opt}} \approx 7.59/\Omega\) and (f) \(T =T_6^{\text{opt}} \approx 19.85/\Omega\).}
    \label{fig:lorentzian}
\end{figure}

\emph{BTL channel with Lorentzian transmission} \textemdash
We will provide two examples of BTL pure-loss channels -- with Lorentzian and box transmission spectra -- which have analytical solutions of total quantum capacity for optimized input/output modes.
For a Lorentzian transmission \(\eta(\omega) = \eta_{\text{max}}\frac{\kappa^2}{\omega^2 + \kappa^2}\), the optimal mode profiles are sine and cosine functions restricted to the time interval $[-T/2,T/2]$ up to normalization. More specifically, the eigenvalues $\lambda_n$ ($n=1,2,...$) satisfy $\tan{(C_n\kappa T)} = 2C_n / (C_n^2-1)$ where $C_n = \sqrt{(\eta_{\text{max}}-\lambda_n)/\lambda_n}$, with the eigenfunctions (see Appendix~\ref{appendix:lorentzian})
\begin{equation}
    \begin{split}
        f_n(t) =& 
        \begin{cases}
            \cos{(C_n\kappa t)} & n \text{ is odd} \\ \sin{(C_n\kappa t)} & n \text{ is even}
        \end{cases}
        , \quad 0 \leq C_n < 1 , \\
        f_n(t) =& 
        \begin{cases}
            \cos{(C_n\kappa t)} & n \text{ is even} \\ \sin{(C_n\kappa t)} & n \text{ is odd}
        \end{cases}
        , \quad C_n > 1 .
    \end{split}
\end{equation}
Notice that the eigenfunctions $f_n(t)$ are generally discontinuous at $t=\pm T/2$ (Fig.~\ref{fig:lorentzian}(d)).

We plot the optimal transmissivities $\lambda_n$ and the resulting time-limited quantum capacity $Q(T)$ (Eq.~(\ref{eq:time_limited_capacity})) for different values of $T$ in Fig.~\ref{fig:lorentzian}(a-b).
In the longtime limit, the first few transmissivities $\lambda_n$ approach $\eta_{\text{max}}$ and $Q(T)$ approaches the maximum capacity $Q^{\text{max}}$.
Interestingly, the time-limited quantum capacity rate $Q(T)$ has multiple discontinuities corresponding to one of the transmissivities crossing 0.5.
This is because the pure-loss capacity $q(\eta)$ (Eq.~(\ref{eq:pure_loss_capacity})) is discontinuous at $\eta=0.5$.
Physically, as we increase the signal duration $T$, every discontinuity in $Q(T)$ represents a non-trivial quantum channel opening up for quantum communication, since $q(\eta)=0$ for $\eta \leq 0.5$.
For the Lorentzian transmission with $\eta_{\text{max}}=1$, the $n$th quantum channel is opened up precisely at $T=(2n-1)\pi/2\kappa, n=1,2,...$ (see Appendix~\ref{appendix:lorentzian}).

In many situations, it may be practically beneficial to only use at most a small, fixed number of modes for encoding input signals. An important conclusion from our results is the existence of an optimal signal duration $T$ when only a few quantum channels are in use.
While the total quantum capacity $Q(T)\times T$ is always increasing in time (Fig.~\ref{fig:lorentzian}(c)), to saturate it we need to use all the available channels.
We define \(Q_{k}(T)=\frac{1}{T} \sum_{j=1}^{k} q\left(\lambda_j\right)\) as the contribution to the time-limited quantum capacity rate $Q(T)$ from the first \(k\) channels with highest transmissivities.
For large $T$, the optimal modes are narrow in frequency whose transmissivities approach $\eta_{\text{max}}$ (Fig.~\ref{fig:lorentzian}(a)), leading to saturation of the quantum capacity $Q_{k}(T) \times T$ (Fig.~\ref{fig:lorentzian}(c)) and reduced information rate $Q_{k}(T)$.
From Fig.~\ref{fig:lorentzian}(b), we find that the information rate $Q_{k}(T)$ for using $k$ channels is maximized at some optimal signal duration \(T_k^{\text{opt}}\).
The optimal duration \(T_k^{\text{opt}}\) and mode profiles for $k=4$ are shown in Fig.~\ref{fig:lorentzian}(d).
Intuitively, the quantum capacity is bounded by both the signal bandwidth and the channel bandwidth (see Appendix~\ref{section:bound}).

\emph{BTL channel with box transmission} \textemdash
Another example that we can solve analytically is a box transmission with $\eta(\omega)=\bar{\eta}$ for $\omega \in [-\Omega,\Omega]$ and 0 otherwise.
The optimal mode profiles are $f_n(t) \propto \psi_n(c,t), t\in [-T/2,T/2]$ with transmissivities $\lambda_n = \bar{\eta} \lambda^s_n(c)$, where $c = \Omega T /2$ and
\begin{align}
    \lambda^s_n(c) & = \frac{2c}{\pi}\left(R_{0n}(c,1)^2\right) \nonumber \\ \psi_n(c,t)  & = \frac{\sqrt{2\lambda^s_n(c)/T}}{\mu_n(c)}S_{0n}\left(c,2t/T\right) .
\end{align}
Here \(R_{n,m}(\gamma,z)\) and \(S_{n,m}(\gamma,z)\) are the radial and angular spheroidal functions of the first kind respectively with \(\mu_n(c) = \sqrt{\int_{-1}^1S_{0n}(c,t)^2 \ dt}\).
The bandlimited functions $\psi_n(c,t)$ are the prolate spheroidal functions of zeroth order (also known as \textit{Slepian functions})~\cite{moore_prolate_2004}, which have also been useful for recent applications in quantum metrology \cite{norris2018optimally,rudnicki2024spectral}.
We plot the optimal mode profiles at \(T = T_2^{\text{opt}}\) and \(T =T_6^{\text{opt}}\) in Fig.~\ref{fig:lorentzian} (e-f). Notice that for short $T$ the eigenfunctions are discontinuous at $t=\pm T/2$, while for long $T$, the Slepian functions can smoothly go to zero.
For more results on box transmission, see Appendix~\ref{appendix:box}.

\begin{figure}
    \centering
    \includegraphics[width=0.5\textwidth]{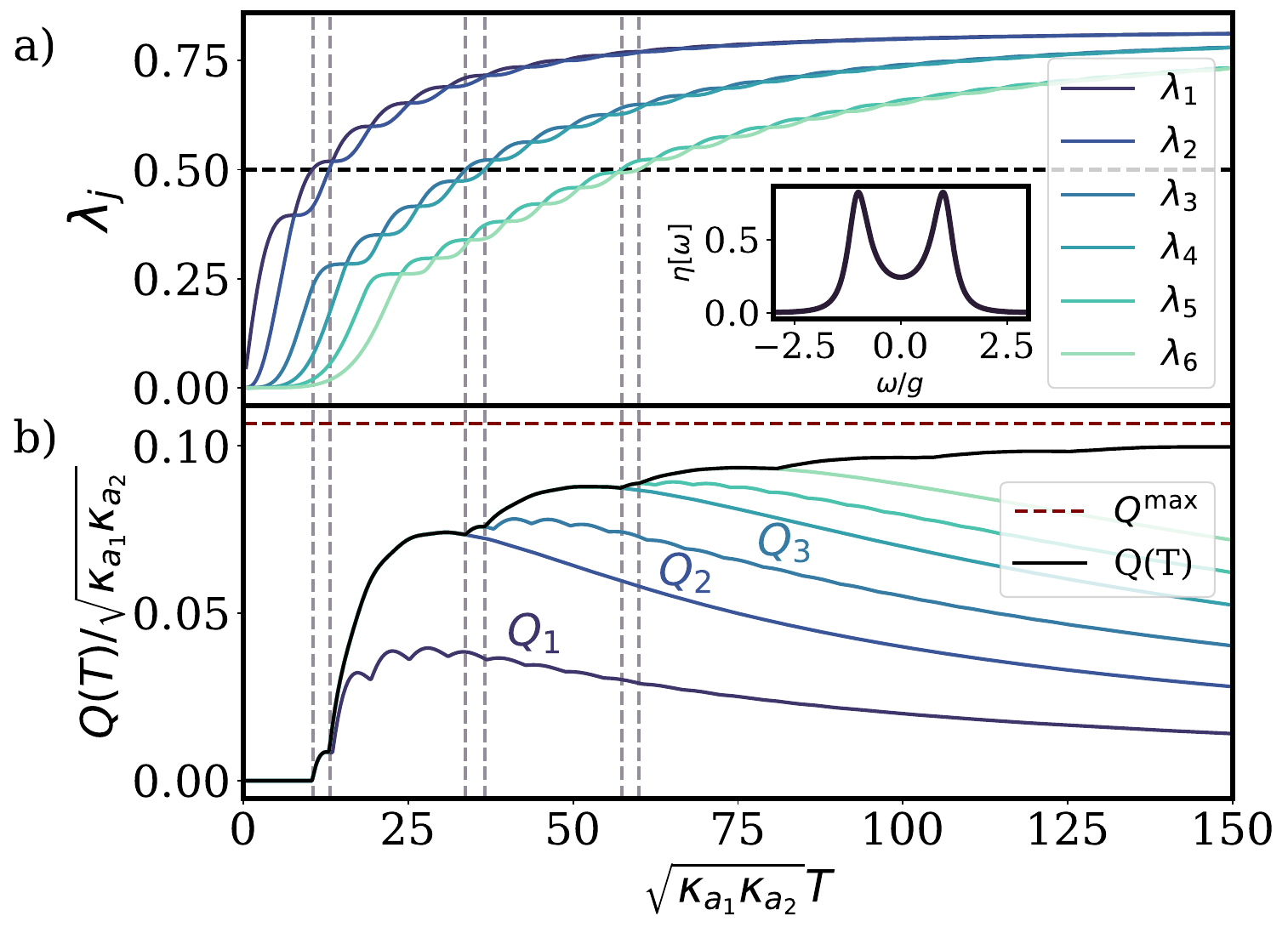}
    \caption{We consider an opto-mechanical transducer with the following parameters:
    \(\kappa_{a_1}/g = 0.5, \kappa_{a_2}/g = 7, \kappa_b/g = 0.1\).
    (a) Transmission eigenvalues for different signal duration $T$. Inset in (a) is the transmission spectrum. (b) Time-limited quantum capacity rate $Q(T)$ (black). The colored lines represent the contributions $Q_k(T)$ of the first \(k\) channels. The vertical gray lines correspond to times at which a new quantum channel opens.}
    \label{fig:om-transducer}
\end{figure}

\emph{BTL channel with generic transmission} \textemdash
While the previous two examples provided many insights, we can also apply our technique to realistic scenarios where analytic solutions may not be found.
Consider an opto-mechanical transducer which interconverts microwave and optical quantum signals $\opa_1$ and $\opa_2$ through an intermediate mechanical mode $\opb$.
The Hamiltonian in the rotating frame is $\oph = g (\opa_1 + \opa_2) \opbd + h.c.$, where the mechanical mode $\opb$ has loss rate $\kappa_b$ and modes $\opa_1,\opa_2$ couple to their input-output waveguides with external coupling rates $\kappa_{a_1},\kappa_{a_2}$; for simplicity, we also take the limit where all internal losses are negligible.
The transmission spectrum from the input of $\opa_1$ to the output of $\opa_2$ is~\cite{wang_quantum_2022}
\begin{equation}
    \eta[\omega] = \left|\frac{\sqrt{\kappa_{a_1}\kappa_{a_2}}g^2}{D[\omega]}\right|^2 , \quad 
    D[\omega] = \det \begin{pmatrix}
    \chi_{a_1}^{-1} & ig & 0 \\ 
    ig & \chi_b^{-1} & ig \\ 
    0 & ig & \chi_{a_2}^{-1}
    \end{pmatrix} ,
\end{equation}
where \(\chi_c^{-1} = i\omega + \kappa_c/2\). 

As an illustrative example, we pick strong-coupling parameters that yield a non-trivial double-peaked transmission spectrum, with each peak rising above the threshold value of 0.5 (inset in Fig.~\ref{fig:om-transducer}(a)).
There are degeneracies in the eigenvalues at various times (Fig.~\ref{fig:om-transducer}(a)), which leads to non-smooth \(Q_{k}(T)\) curves (Fig.~\ref{fig:om-transducer}(b)) with discontinuities corresponding to degenerate points in the eigenvalue spectrum. 
The times at which a new quantum channel opens up (vertical dashed lines in Fig.~\ref{fig:om-transducer}) are not evenly spaced as in the Lorentzian case.
We discuss an parameter set in Appendix~\ref{appendix:om-transducer}.

\emph{Discussion}\textemdash
In summary, we have considered quantum communication over finite bandwidth and established the optimal communication strategies when the input quantum signals are time-limited.
Specifically, the optimal input mode profiles and transmissivities can be obtained by solving an eigenvalue problem.
We found analytical solutions for Lorentzian and box transmission spectra, along with numerical solutions for various other transmissions.
Our findings reveal a general feature of sequential activation of quantum channels as the input signal duration increases, as well as the existence of optimal signal duration for scenarios where only a limited number of channels are in use.

In future works, it would be interesting to employ the optimized signal modes in near-term quantum state transfer experiments. We could extend to more general bosonic quantum channels with amplification beyond the passive Gaussian channel considered here.
Moreover, our general formalism enables the exploration of constraints other than time limitations in various wave transport phenomena.

\begin{acknowledgments}
    We acknowledge support from the ARO(W911NF-23-1-0077), ARO MURI (W911NF-21-1-0325), AFOSR MURI (FA9550-19-1-0399, FA9550-21-1-0209, FA9550-23-1-0338), DARPA (HR0011-24-9-0359, HR0011-24-9-0361), NSF (OMA-1936118, ERC-1941583, OMA-2137642, OSI-2326767, CCF-2312755), NTT Research, Packard Foundation (2020-71479), and the Marshall and Arlene Bennett Family Research Program. This material is based upon work supported by the U.S. Department of Energy, Office of Science, National Quantum Information Science Research Centers and Advanced Scientific Computing Research (ASCR) program under contract number DE-AC02-06CH11357 as part of the InterQnet quantum networking project.
\end{acknowledgments}

\appendix

\section{Proof of Theorem~\ref{theorem}}
\label{appendix:proof}
For any input mode profiles $\{f_k\}$ and the readout mode profiles $\{h_k\}$, the singular value decomposition of the scattering matrix $S=V^\dagger \Sigma U$ always exists. Our goal here is to find the optimal setting, i.e., the optimal linear combination of the input modes and the optimal readout modes, that maximizes the singular values $\{ \sigma_1,\ldots,\sigma_n \}$ of $S$ for a given set of input mode profiles $\{f_k\}$.

\subsubsection{Case 1: $\text{span} \{h_k\} = \text{span} \{g_k\}$}
In this case, each $g_k$ can be written as a linear combination of $h_k$. After the singular value decomposition we have $\ave{g_i',h_j'} \propto \delta_{ij}$ which leads to $g_i' \propto h_i'$. Therefore, after the transformation $\{g_k'\}$ becomes mutually orthogonal, which is the optimal setting in Theorem~\ref{theorem}.

\subsubsection{Case 2: $\text{span} \{h_k\} \neq \text{span} \{g_k\}$}
In this case, we can first transform $h_k$ via $V$, which does not change the singular values of $S$, such that $\{h_1,...,h_m\} \subset \text{span} \{g_k\}$ while $\{h_{m+1},...,h_n\}$ are orthogonal to $\text{span} \{g_k\}$. Now $S$ becomes
\begin{equation}
    S = \left( A_{n \times m}, 0_{n \times (n-m)} \right) ,
\end{equation}
where the eigenvalues of $A^\dagger A$ are $\mu_k = \sigma_k^2,1\leq k \leq m$ and $\sigma_{k>m}=0$.

We can find another choice of measurement basis by complementing $\{h_1,...,h_m\}$ with $\{\tilde{h}_{m+1},...,\tilde{h}_n\}$ such that $\text{span} \{h_1,...,h_m, \tilde{h}_{m+1},...,\tilde{h}_n\} = \text{span} \{g_k\}$. The leads to a modified matrix
\begin{equation}
    \tilde{S} = \left( A_{n \times m}, B_{n \times (n-m)} \right) ,
\end{equation}
where $B$ is non-zero. Notice that
\begin{equation}
    \tilde{S}^\dagger \tilde{S} = 
    \begin{pmatrix}
    A^\dagger A & A^\dagger B \\
    B^\dagger A & B^\dagger B
    \end{pmatrix} ,
\end{equation}
whose eigenvalues are $\lambda_k = \tilde{\sigma}_k^2, k = 1,...,n$ where $\tilde{\sigma}_k$ are the singular values of $\tilde{S}$.

Assuming the eigenvalues are in descending order, i.e., $\lambda_1 \geq ... \geq \lambda_n$ and $\mu_1 \geq ... \geq \mu_m$. From Cauchy’s interlacing theorem~\cite{zhang2011}, we have
\begin{equation}
    \lambda_k \geq \mu_k \geq \lambda_{n-m+k}, \quad k = 1,...,m .
\end{equation}
Therefore $\tilde{S}$ has larger singular values than $S$ which completes our proof.

\section{Derivation of Eq.~(\ref{eigProblem})}
\label{appendix:derivation}
From Theorem~\ref{theorem}, The optimal profiles $\{f_i\}$ satisfy $\ave{g_i,g_j} = \lambda_i \delta_{ij}$. Since
\begin{equation}
    \begin{split}
        f_i(\omega) =& \frac{1}{\sqrt{2\pi}} \int_{-\infty}^{\infty} f_i(t) e^{-i\omega t} \text{d} t \\
        =& \frac{1}{\sqrt{2\pi}} \int_{-T/2}^{T/2} f_i(t) e^{-i\omega t} \text{d} t ,
    \end{split}
\end{equation}
we have
\begin{equation}
    \begin{split}
        & \ave{g_i,g_j} = \lambda_i \delta_{ij} \\
        =& \int_{-\infty}^{\infty} \eta(\omega) f_i^*(\omega) f_j(\omega) \text{d} \omega \\
        =& \frac{1}{2\pi} \int_{-\infty}^{\infty} \eta(\omega) \text{d} \omega \iint_{-T/2}^{T/2} f_i^*(t) f_j(t') e^{i\omega (t-t')} \text{d} t \text{d} t' \\
        =& \iint_{-T/2}^{T/2} f_i^*(t) M(t,t') f_j (t') \text{d} t \text{d} t' ,
    \end{split}
\end{equation}
where
\begin{equation}
    M(t,t') = \frac{1}{2\pi} \int_{-\infty}^{\infty} \eta(\omega) e^{i\omega (t-t')} \text{d} \omega = \tilde{\eta} (t-t') 
\end{equation}
is a Hermitian linear operator, i.e., $M^*(t',t)=M(t,t')$.
Therefore, solving $\{f_i\}$ is equivalent to solving the eigenfunction Eq.~(\ref{eigProblem}).

\section{Convergence of \(Q(T)\) to \(Q^{\text{max}}\)}
\label{appendix:convergence}
We can rewrite the eigenvalue problem in Eq.~(\ref{eigProblem}) as 
\begin{equation}
\label{eq:appEigEq}
    \underbrace{\int_{-\infty}^{\infty} dt'\ \tilde{\eta}(t-t') D_{T/2}(t') f_k(t')}_{\equiv F_T(t)} = \lambda_k f_k(t)
\end{equation}
where we define the window function
\begin{equation}
    D_{\tau}(t) = \begin{cases} 1 & |t|\leq \tau \\ 0 & |t| > \tau \end{cases}
\end{equation}

We would like to write Eq.~(\ref{eq:appEigEq}) in frequency space for large \(T\). However, it is only defined for \(t \in [-T/2,T/2]\), and as such, we can write
\begin{equation}
    \lim_{T\to \infty} \int_{-T/2}^{T/2} dt\ e^{-i\omega t} F_T(t) =  \lambda_k \lim_{T\to \infty} \int_{-T/2}^{T/2} dt \ e^{-i\omega t} f_k(t)
\end{equation}
In this limit, notice that \(D_{T/2}(t') f(t') \to f(t')\), since the window becomes much wider than \(f(t)\). Then, the integrals over \(t\) become well-defined Fourier transforms, and in the limit \(T \to \infty\) we have
\begin{equation}
    \mathcal{F}\left[ \int_{-\infty}^{\infty} dt'\ \tilde{\eta}(t-t')  f_k(t')\right] = \lambda_k f_k(\omega) 
\end{equation}
where \(\mathcal{F}[]\) represents the Fourier transform. Now, we can apply convolution theorem on the left hand side, which yields
\begin{equation}
    \eta(\omega) f_k(\omega) = \lambda_k f_k(\omega).
\end{equation}

This is trivially solved by the eigenfunctions \(f_k(\omega) = \delta(\omega - k \Delta \omega)\), for some sampling rate \(\Delta \omega = 2\pi/T\), yielding eigenvalues \(\lambda_k \to \eta(k \Delta\omega)\). This is exactly the zero linewidth case in which \(Q^{\text{max}}\) is defined, where each frequency is its own mode. Then, we can write 
\begin{equation}
    \lim_{T \to \infty} Q(T) = \lim_{T \to \infty} \frac{1}{2\pi} \Delta \omega \sum_k q_1\left(\eta(k \Delta\omega)\right) = Q^{\text{max}}.
\end{equation}

\section{Analytic solution for Lorentzian transmission}
\label{appendix:lorentzian}
For the special scenario where the transmission spectrum is a Lorentzian, we can find analytical solutions of the eigenfunctions, and thus the optimal basis functions. For a cavity symmetrically coupled to two waveguides, we have
\begin{align}
    t(\omega)  = -\frac{\kappa}{i\omega + \kappa} \quad & \Rightarrow \quad \eta(\omega) = \frac{\kappa^2}{\omega^2 + \kappa^2} \nonumber \\ 
   &  \Rightarrow \quad \tilde{\eta} (t) = \frac{\kappa}{2} e^{-\kappa |t|} 
\end{align}
The eigenequation becomes
\begin{equation}
    \begin{split}
        \lambda f(t) =& \frac{\kappa}{2} \int_{-T/2}^{T/2} e^{-\kappa |t-t'|} f(t') d t' \\
        =& \frac{\kappa}{2} \int_{-T/2}^{t} e^{-\kappa (t-t')} f(t') d t' \\
        & + \frac{\kappa}{2} \int_{t}^{T/2} e^{-\kappa (t'-t)} f(t') d t' \\
    \end{split}
\end{equation}

\subsection{Eigenfunctions and eigenvalues}
For the integral equation above, take the first and second derivatives gives
\begin{equation}
    \begin{split}
        \lambda f'(t) =& \frac{\kappa}{2} (-\kappa) \int_{-T/2}^{t} e^{-\kappa (t-t')} f(t') d t'  \\ & + \frac{\kappa}{2} \kappa \int_{t}^{T/2} e^{-\kappa (t'-t)} f(t') d t' \\
        \lambda f''(t) =& \frac{\kappa}{2} (-\kappa)^2 \int_{-T/2}^{t} e^{-\kappa (t-t')} f(t') d t' \\ & + \frac{\kappa}{2} \kappa^2 \int_{t}^{T/2} e^{-\kappa (t'-t)} f(t') d t' - \kappa^2 f(t)\\ =& \kappa^2 (\lambda - 1) f(t) 
    \end{split}
\end{equation}
Therefore the eigenfunctions take the form
\begin{equation}
    f(t) = A e^{i C \kappa t} + B e^{-i C \kappa t}, \quad t \in [-T/2,T/2]
\end{equation}
where
\begin{equation}
    C=\sqrt{\frac{1-\lambda}{\lambda}} \geq 0 .
\end{equation}

The original eigenequation now becomes
\begin{equation}
    \begin{split}
        & A \left[ (-1+iC) e^{-\frac{1}{2} (1+iC)\kappa T} e^{-\kappa t} - (1+iC) e^{-\frac{1}{2} (1-iC)\kappa T} e^{\kappa t} \right] \\
        +& B \left[ (-1+iC) e^{-\frac{1}{2} (1+iC)\kappa T} e^{\kappa t} - (1+iC) e^{-\frac{1}{2} (1-iC)\kappa T} e^{-\kappa t} \right] \\
        & = 0 .
    \end{split}
\end{equation}
This gives $A = \pm B$ and
\begin{equation}
    e^{iC\kappa T} = \pm \frac{C^2-1 + 2Ci}{C^2+1} .
\end{equation}
Equivalently, $C$ satisfies
\begin{equation}
    \tan C \kappa T = \frac{2C}{C^2-1} ,
\end{equation}
and the parity of the eigenfunction depends on the value of $C$.

We can identify the intervals within which there is exactly one solution of $C$, and the corresponding eigenfunction. For $0 \leq C < 1$, the intervals are
\begin{equation}
    C \in \left( \frac{(2n-1)\pi}{2\kappa T}, \frac{n\pi}{\kappa T} \right), \quad \left \{ 
    \begin{aligned}
    f(t) = \cos(C \kappa t) & \quad 2 \nmid n \\
    f(t) = \sin(C \kappa t) & \quad 2 \mid n
    \end{aligned}
    \right .
\end{equation}
with \(1 \leq n < \frac{\kappa T}{\pi} + \frac{1}{2}\). 
For $C>1$, the intervals are
\begin{equation}
    C \in \left( \frac{n\pi}{\kappa T}, \frac{(2n+1)\pi}{2\kappa T} \right), \quad \left \{ 
    \begin{aligned}
    f(t) = \cos(C \kappa t) & \quad 2 \mid n \\
    f(t) = \sin(C \kappa t) & \quad 2 \nmid n
    \end{aligned}
    \right . 
\end{equation}
with \(n > \frac{\kappa T}{\pi} - \frac{1}{2}\).

\subsection{Opening up of a new quantum channel}
The threshold for transmitting quantum information is $\lambda = 0.5$. Notice that for $\lambda > 0.5$, we have $C<1$. Therefore as we increase $T$, the $n$th quantum channel is opened up precisely at
\begin{equation}
    T = \frac{(2n-1)\pi}{2\kappa} , \quad n = 1,2,...
\end{equation}

\subsection{Convergence to the total quantum capacity} 
When $T\rightarrow \infty$, the contributions from all eigenfunctions in the time domain should converge to the total quantum capacity calculated in the frequency domain. In the limit of $T\rightarrow \infty$, the solutions where $C<1$ are $C_n \approx n\pi/\kappa T$ where $1 \leq n < \kappa T/\pi$. Therefore the total quantum capacity over time $T$ is
\begin{equation}
    \begin{split}
        Q(T) =& \sum_{n=1}^{\kappa T/\pi} \log \frac{\lambda_n}{1-\lambda_n} = -2 \sum_{n=1}^{\kappa T/\pi} \log C_n \\
        \approx & 2 \sum_{n=1}^{\kappa T/\pi} \log \frac{\kappa T}{n \pi} \approx 2 \int_{1}^{\kappa T/\pi} \log \frac{\kappa T}{n \pi} dn \\
        \stackrel{\omega= n\pi /T}{=} & \frac{2T}{\pi} \int_{\pi/T}^{\kappa} \log \frac{\kappa}{\omega} d\omega \\
        \approx & \frac{2T}{\pi} \int_{0}^{\kappa} \log \frac{\kappa}{\omega} d\omega .
    \end{split}
\end{equation}
In the frequency domain, the total quantum capacity per unit time is
\begin{equation}
\begin{split}
    Q & = \frac{1}{2\pi} \int_{-\infty}^{\infty} \max\left(0, \log \frac{\eta(\omega)}{1-\eta(\omega)} \right) d\omega \\ & = \frac{1}{2\pi} \int_{-\infty}^{\infty} \max\left(0, \log \frac{\kappa^2}{\omega^2} \right) d\omega \\ & = \frac{2}{\pi} \int_{0}^{\kappa} \log \frac{\kappa}{\omega} d\omega .
\end{split}
\end{equation}
Therefore, the total quantum capacities calculated in the time and frequency domains indeed agree with each other.

\section{Results on box transmission}
\label{appendix:box}
We can rewrite Eq.~(\ref{eq:appEigEq}) in the Fourier domain using convolution theorem as
\begin{equation}
\label{eigProblemFreq}
    \eta(\omega) \int d\omega' \frac{\sin(\omega'T/2)}{\omega'} f_k(\omega-\omega') = \lambda_k f_k(\omega). 
\end{equation}

For a box transmission of height \(\bar{\eta}\) and width \(2\Omega\) as described in the main text,  Eq.~(\ref{eigProblemFreq})) becomes 
\begin{equation}
     \bar{\eta} \int_{-\Omega}^{\Omega} d\omega' \frac{\sin(\omega'T/2)}{\omega'} \tilde{f}_k(\omega-\omega') = \lambda_k f_k(\omega). 
\end{equation}

For a constant \(\bar{\eta}\), this is exactly the Slepian eigenvalue equation in~\cite{moore_prolate_2004}. This yields the Slepian solution discussed in the main text. In Fig.~\ref{fig:boxAppendix}, we plot the \(Q(T)\) and \(Q(T)\times T\) curves for the box transmission, as we do in Fig.~\ref{fig:lorentzian}(b-c). Notice the offset of the \(Q^{\text{max}}\times T\) curve from \(Q(T)\times T\) in Fig.~\ref{fig:boxAppendix}(b) - this is a feature of Riemann error due to discretization in our numerical calculations. 
 
\begin{figure}
    \centering
    \includegraphics[width=0.5\textwidth]{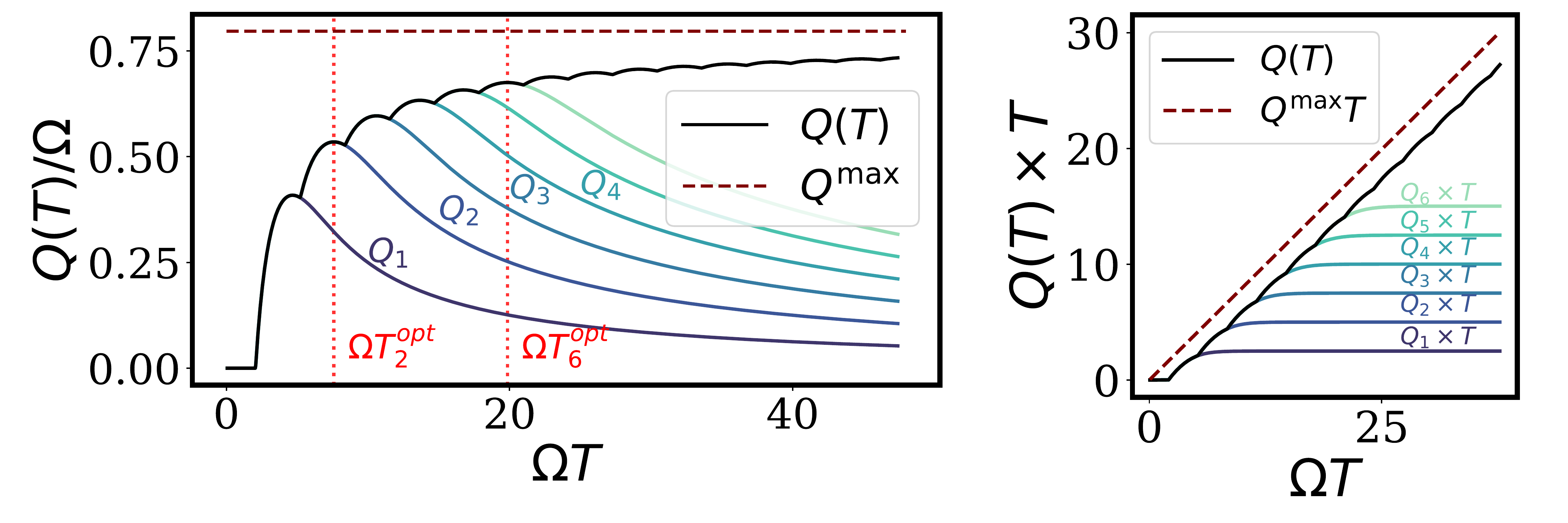}
    \caption{a) \(Q(T)\) curve (black) and \(Q_k(T)\) curves (colored) for the a box transmission with height \(\bar{\eta} = 0.85\) as described in the main text. Vertical red lines indicate the optimal times \(T_2^{\text{opt}}\) and \(T_6^{\text{opt}}\) with corresponding input modes plotted in Fig.~\ref{fig:lorentzian}(e-f). b) Corresponding \(Q(T) \times T\) curve (black) and \(Q_k(T)\times T\) (colored) curves. We also plot \(Q^{\text{max}}T\) (maroon dashed line).}
    \label{fig:boxAppendix}
\end{figure}

\section{Lower time bound for non-trivial capacity}
\label{section:bound}
We discuss a general upper bound on the highest possible transmission \(\lambda_1\) for some given time \(T\), which in turn allows us to provide a lower bound for the time at which \(Q(T)\) is first non-zero. We know from Theorem~\ref{theorem} that the maximum transmissivity \(\lambda_1\) is given by 
\begin{equation}
    \lambda_1 = |\langle \tau^*(\omega) f_1(\omega), \tau^*(\omega)f_1(\omega)\rangle|^2 = \mathbb{E}(\eta(\omega))
\end{equation}
where we assume \(f_1(\omega)\) is the input mode corresponding to the largest transmissivity for time \(T\). We can write \(\lambda_1\) as an expectation value of \(\eta(\omega)\) in a probability distribution \(P(\omega) = |f_1(\omega)|^2\) since \(f_1(\omega)\) is normalized. Then, assuming the maximum value of the spectrum \(\eta(\omega)\) lies in some range \(\omega \in [-\Omega,\Omega]\), we can trivially write
\begin{equation}
    \lambda_1 \leq \eta_{\text{max}}\int_{|\omega|<\Omega} P(\omega)d\omega + \int_{|\omega| > \Omega} \eta(\omega) P(\omega) d\omega
\end{equation}
Now, Chebyshev's inequality \cite{lin_probability_2010} helps bound tails of probability distributions as 
\begin{equation}
    \int_{|\omega| > l\sigma_{\omega}} P(\omega) \leq \frac{1}{l^2}
\end{equation}
where \(\sigma_{\omega}^2\) is the variance of the probability distribution \(P(\omega)\). Picking \(\Omega = l \sigma_{\omega}\), we can bound
\begin{align}
    \int_{|\omega| > \Omega} \eta(\omega) P(\omega) d\omega &  \leq \left(\text{max}_{|\omega| > \Omega} \eta(\omega)\right) \int_{|\omega| >\Omega} P(\omega) d\omega  \nonumber \\ & \leq \left(\text{max}_{|\omega| > \Omega} \eta(\omega)\right) \frac{\sigma_\omega^2}{\Omega^2}
\end{align}

We know that for a time-limited input function defined in the interval \(t \in [-T/2,T/2]\), the maximum standard deviation is \(\sigma_t = T/2\). Then, by the uncertainty principle for signal processing, \(\sigma_t \sigma_{\omega} \geq \frac{1}{2\pi}\), we get that \( \sigma_{\omega}^2 \geq 4\pi^2/T^2\). If we then also assume \(\eta(\omega)\) is monotonically decreasing for \(|\omega|>\Omega\), we can write 
\begin{equation}
\label{eq:bound}
    \lambda_1 \leq  \text{min}_{\Omega} \left[ \ \eta_{\text{max}}\int_{|\omega|<\Omega} P(\omega)d\omega+\eta(\Omega) \frac{4\pi^2}{\Omega^2T^2}\right].
\end{equation}

Note that there is a freedom of choice in \(\Omega\) which we can minimize to get the tightest bound. This bound gives us valuable intuition for the capacities of certain input modes. Critically, the ratio of the variances of the probability distribution \(|f_k(\omega)|^2\) and transmission \(\eta(\omega)\) determines the transmission \(\lambda_k\). 

The total time \(T\) determines a minimum variance on the probability distribution \(\sigma_{\omega}^2\). If we have \(\Omega T \gg 1\), \(\sigma_{\omega}\) is small compared to the spread of \(\eta(\omega)\), and the input mode approaches a delta function. In this limit, we can find some \(\Omega\) that will have \(\int_{|\omega|<\Omega} P[\omega]d\omega \to 1\), and the second term on the right in Eq.~(\ref{eq:bound}) will go to zero. Then, we see \(\lambda_1 \to \eta_{\text{max}}\), which is what we expect in the \(T \to \infty\) limit since each frequency becomes an independent mode and the quantum capacity approaches the limiting value \(Q^{\text{max}}\). 

In the limit \(\Omega T \ll 1\), we have that \(\sigma_{\omega}^2\) is very large as compared to the spread of \(\eta(\omega)\). Then, we can always find an \(\Omega\) such that \(\eta(\Omega) \ll 1\) while \(\int_{|\omega|<\Omega} P(\omega)d\omega < 1/2\), which would bound \(\lambda_1 < 1/2\), and we would see no quantum capacity. 

\begin{figure}
    \centering
    \includegraphics[width=0.5\textwidth]{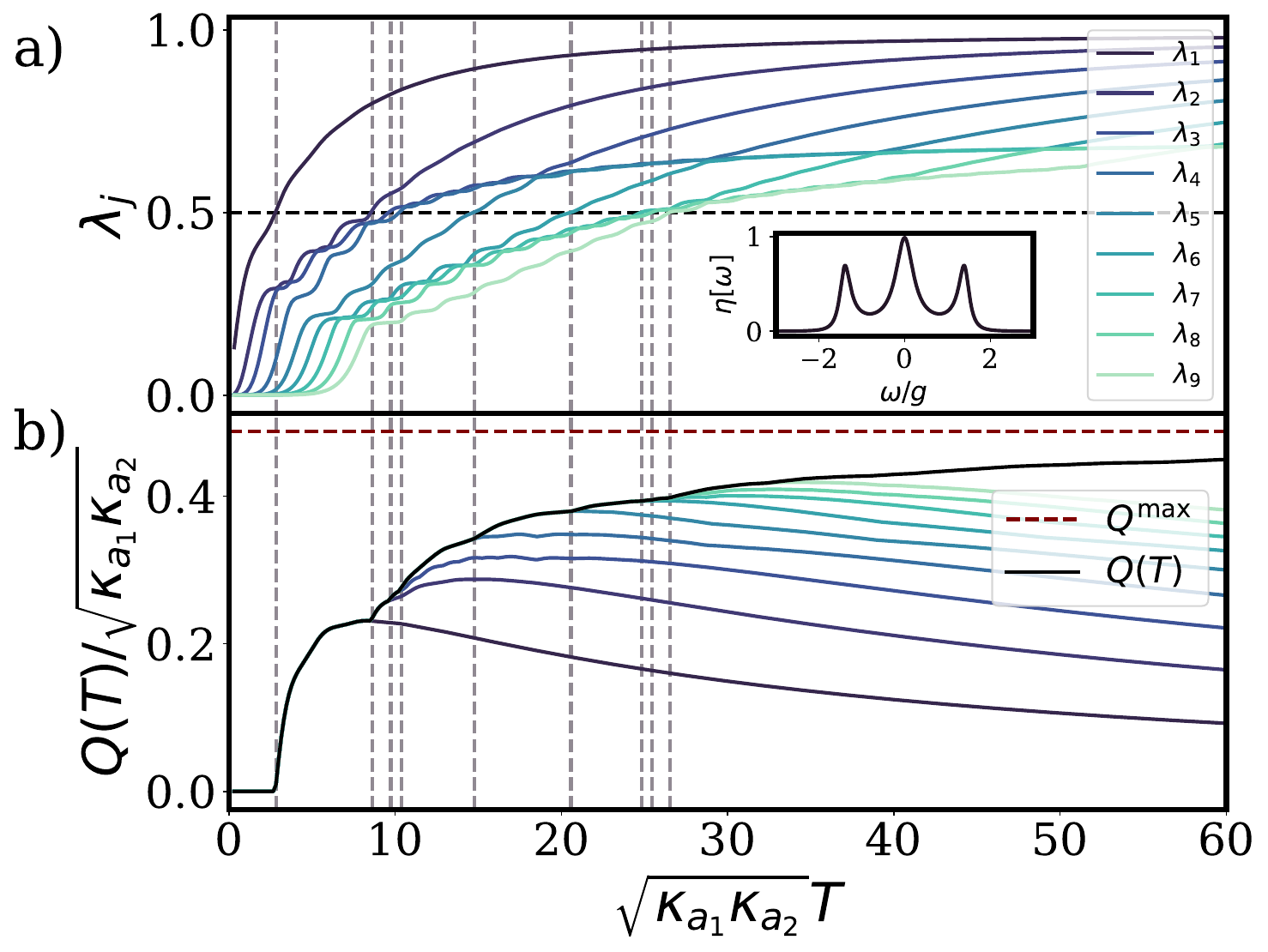}
    \caption{Eigenvalues and \(Q(T)\) for an opto-mechanical transducer for Parameter Set 2. The corresponding transmission spectrum can be seen as an inset in the top panel.}
    \label{fig:PS3}
\end{figure}

\section{Further parameter sets for opto-mechanical transducer}
\label{appendix:om-transducer}
We consider an additional set of parameters for the opto-mechanical transducer system discussed in the main text (see Fig.~\ref{fig:om-transducer}). In particular, this parameter set is given by the parameters \(\kappa_{a_1}/g = \kappa_{a_1}/g  = 0.5\) and \(\kappa_b/g = 0.1\). The transmission has three peaks, typical of such a spectrum. Notably, the \(Q_{k}(T)\) are far more complicated since there is a large number of degeneracies in the eigenvalues \(\lambda_j\) (see Fig.~\ref{fig:PS3}). This also leads to the times for new channels to open to be non-trivial and not evenly spaced as well. 

\bibliography{paper-citations-1,paper-citations-2}

\end{document}